\documentclass[journal]{IEEEtran}
\usepackage{amssymb,amsfonts}
\usepackage[cmex10]{amsmath}
\usepackage{amsthm}
\usepackage{graphicx}
\usepackage[all,cmtip]{xy}
\usepackage{tikz}
\usepackage[colorlinks, citecolor=blue,]{hyperref}
\usepackage{cite}
\usepackage[numbers,sort&compress]{natbib}
\usepackage{subfigure}
\usepackage{titling}
\def\bb#1{{\mathbb #1}}
\newcommand{\R}{{\rm I\!R}}
\newcommand{\dfb}{\stackrel{\Delta}{=}}

\def\scr#1{{\mathcal #1}}

\newcommand{\diag}{{\rm diag\;}}

    \newtheorem{remark}{Remark}
    \newtheorem{lemma}{Lemma}
    \newtheorem{theorem}{Theorem}
    
    \newtheorem{assumption}{Assumption}

\begin{document}
%
\title{Distributed Global Output-Feedback Control for a Class of Euler-Lagrange Systems}
%
%
%

\author{Qingkai Yang, Hao Fang, Jie Chen, Zhong-Ping Jiang, and Ming Cao} 
\date{}
%
%
%
	
\maketitle

\begin{abstract}
This paper investigates the distributed tracking control problem for a class of Euler-Lagrange multi-agent systems when the agents can only measure the positions. In this case, the lack of the separation principle and the strong nonlinearity in unmeasurable states pose severe technical challenges to global output-feedback control design. To overcome these difficulties, a global nonsingular coordinate transformation matrix in the upper triangular form is firstly proposed such that the nonlinear dynamic model can be partially linearized with respect to the unmeasurable states. And, a new type of velocity observers is designed to estimate the unmeasurable velocities for each system. Then, based on the outputs of the velocity observers, we propose distributed control laws that enable the coordinated tracking control system to achieve uniform global exponential stability (UGES). Both theoretical analysis and numerical simulations are presented to validate the effectiveness of the proposed control scheme. 
\end{abstract}


%
\IEEEpeerreviewmaketitle

\section{Introduction}
%
%
%
%

%
\IEEEPARstart{R}{ecently}, intensive attention has been paid to distributed control for Euler-Lagrange systems due to its broad applications. Several approaches have been proposed to deal with the distributed tracking problem, see, for example, sliding-mode method \cite{KhXiMa09,MeReMa11}, disturbance observer \cite{YaShQi15}, and extended proportional-integral control scheme \cite{ChFeLi15}. Extension to handle unknown parameter uncertainties can be found in \cite{NuOrBa11}, where an adaptive controller is proposed to synchronize nonidentical Euler-Lagrange systems with communication time delays. Later,  \cite{Wa13} solves the synchronization problem of networked robotic systems with both the kinematic and dynamic uncertainties using passivity theory. It has also been shown that, under a jointly connected switching network topology, leader-following consensus can be achieved for multiple Euler-Lagrange systems by employing adaptive control \cite{CaHu14}, in which various reference signals, such as sinusoidal and ramp signals, generated by an exosystem are considered.

In order to relax the restrictive requirement for full state measurements in designing the controllers in the existing results, some efforts on partial state feedback control have been made. 
 	 To estimate the unmeasurable velocity, observers are constructed by invoking the Immersion and Invariance (I$\&$I) techniques \cite{AsOrVe10,StAaKa11,SaNuKi12}. Through introducing two extra states, some lower dimensional observer is proposed in \cite{StAaKa11} in comparison to that in \cite{AsOrVe10}, and moreover, the explicit expressions of the observer have been given. In \cite{SaNuKi12}, dynamic scaling and high-gain terms have been adopted to perform the Lyapunov stability analysis. 
	  Note that the dynamics of the observers relying on I$\&$I techniques are generally high dimensional and complex. In addition, it is required to find a certain attractive and invariant manifold in the extended state space of the plant and the observer, which will likely increase the computational burden.

  For multi-agent systems, a consensus algorithm using linear observers is first proposed in \cite{Re09}. And, in \cite{MeReCh13}, a distributed control law with time varying control gains is designed to compensate for the lack of neighbors' velocity measurements.  For the distributed tracking problem,  \cite{ZhDuWe15} presents a sliding mode observer-based controller to track the leader with constant velocity in finite time. 
The more challenging problems of tracking a leader with varying velocity have also been investigated in \cite{MeDiJo14,Ch14,Yang14IET}. When only nominal parameters of Euler-Lagrange systems are available, global asymptotic stability can be ensured using continuous control algorithms with adaptive coupling gains \cite{MeDiJo14}. More generally, in cases when we do not have access to any velocity measurement, it is desirable to coordinate the agents using output-feedback strategies. 
However, some drawbacks of the available  results still exist. For example, tracking errors can only be guaranteed to be uniformly ultimately bounded but not converging to zero \cite{Ch14} and the resulted closed-loop system is only locally but not globally stable under the designed control laws \cite{Yang14IET}. For some specific class of nonlinear systems, the global output feedback control problem has been investigated recently in  \cite{LiJi13tac} and \cite{DoHu14tac}, where the cyclic-small-gain approach and distributed internal model have been introduced respectively to achieve global convergence. To deal with the leader's unavailable velocity measurements, distributed observers are designed for second-order agents  in \cite{HoChBu08}.

The goal of this paper is to address the problem of distributed global output-feedback tracking for multiple Euler-Lagrange systems modeling a class of two-link revolute robot manipulators. Up to now, there is no known result for such distributed global tracking algorithms due to several technical challenges. The main difficulties in achieving global stability lie in the quadratic nonlinearities and the cross terms of the unmeasurable velocity states derived from the Coriolis and centrifugal torques. To eliminate such quadratic terms, different state transformation methods have been utilized in \cite{Ba2000,SaOrPa12,JiKa00,DoJiPa04}. In \cite{Ba2000} and \cite{SaOrPa12}, the coordinate transformation strategies are first applied to simplify the nonlinear models, and then controllers are proposed for one-degree-of-freedom Euler-Lagrange systems and underactuated mechanical systems in their Hamiltonian forms, respectively. It should also be noted that both of the techniques in \cite{JiKa00} and \cite{DoJiPa04} pose constraints on the system model, i.e., a class of nonlinear systems that are linear in unmeasured states, globally stabilizable using output feedback \cite{JiKa00}, and with skew-symmetric Coriolis terms \cite{DoJiPa04}. However, the models of the two-link revolute robot manipulators considered in this paper do not possess any of the properties just mentioned that contribute to the simplification of the system model. So all of these approaches cannot be directly applied to the robot manipulators discussed in this paper. 
Inspired by \cite{JiKa00} and \cite{DoJiPa04}, we shall focus on how to partially linearize the dynamics of the robot manipulators through coordinate transformation and state reconstruction. With the help of the model transformation, a distributed velocity observer is proposed,  which enables us to implement the output-feedback control for multiple robot manipulators such that the tracking errors uniformly globally exponentially converge to zero. 

The rest of the paper is organized as follows. Section \ref{sec-02} reviews the system dynamics and presents the method on how to partially linearize the nonlinear system through coordinate transformation. In Section \ref{sec-03}, an observer-based control strategy is proposed based on the partially linearized system. Section \ref{sec-04} gives the main result of this paper, followed by the numerical simulations in Section \ref{sec-05}. Finally, conclusions are provided in Section \ref{sec-06}.

Notations: $|X|$ denotes the determinant of a real square matrix $X$. $\|x \|$ is used to denote the $2$-norm of  a vector $x$. $I_n$ represents the identity matrix with dimension $n$, and $\mathbf{1}_n$ denotes the column vector whose components are all 1. We use $^iX_{jk}$ to denote the $(j,k)$th element of matrix $X_i$.  And $\overline{\lambda}_X$ and $\underline{\lambda}_X$ are the largest and smallest eigenvalues of a real symmetric matrix $X$, respectively.

\section{Partial linearization}\label{sec-02}
In this section, we first briefly introduce the general expression of Euler-Lagrange systems, followed by the specific dynamics of two-link revolute robot manipulators. Then, we present the process removing the cross terms of the velocity states via coordinate transformation. 

\subsection{Dynamics of Robot Manipulator}

We consider here a group of $n$ mechanical robots, each of which is described by a Euler-Lagrange equation as follows:
\begin{equation}  \label{equ01}
M_i(q_i)\ddot{q_i}+C_i(q_i,\dot{q_i})\dot{q_i}+G_i(q_i)=\tau_i, \quad i=1,\cdots, n,  
\end{equation}
where $q_i $ is the vector of the generalized coordinates, $M_i(q_i)$ is the symmetric positive-definite inertia matrix,
$C_i(q_i,\dot{q_i})\dot{q_i} $ is the Coriolis and centrifugal torque, $G_i(q_i)$ is the vector of the gravitational torques,
and $\tau_i $ is the control torque on robot $i$.

The neighbor relationships between the robots are described by a directed graph $\bb G$ with the vertex set $\mathcal{V}=\{1,2,\cdots,n\}$ and the edge set $\mathcal{E}\subseteq \mathcal{V}\times \mathcal{V} $. We use $\scr{A}=[{a_{ij}}]_{n\times n}$ to denote the adjacency matrix, where $a_{ij}=1$ means there is an edge $(j,i)$ between robots $i$ and $j$, and robot $i$ can obtain information from robot $j$, but not vice versa, and $a_{ij}=0$ otherwise. There is one leader robot and the rest are followers. The interaction relationships among the followers and the leader is denoted by the matrix $ B=\diag\{b_1,\cdots, b_n\}$, where $b_i=1$ if the leader is a neighbor of robot $i$, and $b_i=0$ otherwise.  The Laplacian matrix $ L=[l_{ij}]_{n\times n}$ is defined by $l_{ii}=\sum_{j\in \scr N_i} a_{ij}$ and $l_{ij}=-a_{ij}, \, i\neq j$, where $\scr N_i$ denotes the set of neighbors of robot $i$.

It is well known that a wide range of mechanical systems can be represented by Euler-Lagrange equations, such as robot manipulators, mobile robots and rigid bodies. Here,  we focus on a class of two-link revolute robot manipulators, whose dynamics are given by (see \cite{SpHuVi06} for more details)

{\footnotesize  {\setlength \arraycolsep{1pt}  \begin{eqnarray*}
			&M_i(q_i)= \left[ \begin{array}{cc}
				O_{i(1)}+2 O_{i(2)}\cos(q_{i(2)}), & O_{i(3)}+O_{i(2)}\cos(q_{i(2)}) \\
				O_{i(3)}+O_{i(2)}\cos(q_{i(2)}),&O_{i(3)}
			\end{array} \right],  \\
			&C_i(q_i,\dot{q}_i)= \left[ \begin{array}{cc}
				-O_{i(2)}\sin(q_{i(2)})\dot{q}_{i(2)}, & -O_{i(2)}\sin(q_{i(2)})(\dot{q}_{i(1)}+\dot{q}_{i(2)})\\
				O_{i(2)}\sin(q_{i(2)})\dot{q}_{i(1)}, & 0
			\end{array} \right],  \\
			&G_i(q_i)=\left[ \begin{array}{c}
				O_{i(4)}g\cos(q_{i(1)})+O_{i(5)}g\cos(q_{i(1)}+q_{i(2)})\\
				O_{i(5)}g\cos(q_{i(1)}+q_{i(2)})
			\end{array}\right],
		\end{eqnarray*}} }  	
		
		{\noindent
			where $g$ is the acceleration of gravity, $q_i = [q_{i(1)},q_{i(2)}]^T $ represents the joint angles of the two links and $O_i =[O_{i(1)},O_{i(2)},O_{i(3)},O_{i(4)},O_{i(5)}]=[m_1l_{c1}^2+m_{2}(l_1^2+l_{c2}^2)+J_1+J_2,m_2l_1l_{c2},m_2l_{c2}^2+J_2,
			m_1l_{c1}+m_2l_1,m_2l_{c2}]$, in which the variables $m_i, l_i$ and $J_i$ are, respectively, used to denote the masses, the lengths and the moments of inertia of link $i$, and $l_{ci}$ represents the  distance from the previous joint to the center of mass of link $i$, $i=1,2$.  The inertia matrix $M_i (q_i )$ satisfies the following property: for all $q_i \in \R^2$, there exist positive constants $k_m$ and $k_M$ such that $k_m I_2 \leq M_i(q_i) \leq k_M I_2$.}

		\subsection{Coordinate Transformation}
		In order to linearize the quadratic velocity terms in $C_i(q_i,\dot{q}_i)\dot{q}_i$ and to simplify the dynamics model, motivated by  \cite{DoJiPa04}, we introduce the following coordinate transformation
		\begin{equation}\label{equ02}
		z_i=T_i(q_i)\dot{q}_i,
		\end{equation}
		where $T_i(q_i)\in \R^{2\times 2}$, a nonsingular matrix with bounded elements to be determined, is constructed as follows
		\begin{equation}\label{equ2a-01}
		T_i(q_i)= \left[ 
		\begin{array}{cc}
		^{i}T_{11} & ^{i}T_{12} \\
		^{i}T_{21} & ^{i}T_{22}
		\end{array} \right]
		= \left[
		\begin{array}{cc}
		^{i}M_{11} & ^{i}M_{12} \\
		0 & ^{i}T_{22}
		\end{array} \right],
		\end{equation}
		where $^{i}T_{22}$ needs to be determined.
	 Here, instead of fully linearizing system \eqref{equ01}, we aim at partially linearizing the nonlinear mechanical system. Hence, the transformation matrix is chosen to be in its upper triangular form \eqref{equ2a-01}, which not only simplifies the system model, but also reduces the computational complexity greatly  when solving a set of partial differential equations (PDEs).		
		Considering the system model \eqref{equ01}, the dynamics of the new state $z_i$ can be described by 
		\begin{align}\label{equ03}
		\dot{z}_i=&\left(\dot{T_i}(q_i)\dot{q}_i-T_i(q_i)M_i(q_i)^{-1}C_i(q_i,\dot{q}_i)\dot{q}_i\right) \nonumber \\
		&+T_i(q_i)M_i(q_i)^{-1}\left(\tau_i-G_i(q_i)\right).
		\end{align}
		Note that the matrix $T_i(q_i)$ is globally nonsingular  as long as $^{i}T_{22}$ is not equal to zero. In order to determine $^{i}T_{22}$, substituting \eqref{equ2a-01} into \eqref{equ03} yields 
		{\small
			\begin{align}
			\dot{z}_{i(2)} =&  \frac{\partial  ^{i}T_{22}}{\partial q_{i(1)}} \dot{q}_{i(1)}\dot{q}_{i(2)}+\frac{\partial  ^{i}T_{22}}{\partial q_{i(2)}} \dot{q}_{i(2)}^2 \nonumber \\
			&-\frac{^{i}T_{22} \, O_{i(2)}\sin(q_{i(2)})}{|M_i(q_i)| \, ^{i}M_{11}} \left(\, ^{i}M_{11} \dot{q}_{i(1)}+ \,^{i}M_{12} \dot{q}_{i(2)}\right)^2 \nonumber \\
			&+\frac{^{i}T_{22}}{|M_i(q_i)|}O_{i(2)}\sin(q_{i(2)})\frac{\,^{i}M_{12}\left(\,^{i}M_{12}-\,^{i}M_{11} \right) }{ \,^{i}M_{11}} \dot{q}_{i(2)}^2 \nonumber \\
			& + \,^{i}T_{22}(M_i^{-1})_{21} u_{i(1)}+\,^{i}T_{22} (M_i^{-1})_{22} u_{i(2)}. \label{equ10}
			\end{align}}
		Here, for the purpose of removing the cross coupling term $\dot{q}_{i(1)}\dot{q}_{i(2)}$ in $\dot{z}_{i(2)}$, we let 
		\begin{gather}
		\frac{\partial  ^{i}T_{22}}{\partial q_{i(1)}} \dot{q}_{i(1)}\dot{q}_{i(2)}+\frac{\partial  ^{i}T_{22}}{\partial q_{i(2)}} \dot{q}_{i(2)}^2=  \nonumber \\
		\frac{^{i}T_{22}}{|M_i(q_i)|}O_{i(2)}\sin(q_{i(2)})\frac{\,^{i}M_{12}\left(\,^{i}M_{11}-\, ^{i}M_{12} \right) }{ \,^{i}M_{11}} \dot{q}_{i(2)}^2. \label{equ11}
		\end{gather}
		With \eqref{equ11}, the dynamics of $\dot{z}_{i(2)}$ reduce to
		\begin{align}
		\dot{z}_{i(2)} =
		&-\frac{^{i}T_{22} \, O_{i(2)}\sin(q_{i(2)})}{|M_i(q_i)| \, ^{i}M_{11}} \left(\, ^{i}M_{11} \dot{q}_{i(1)}+ \,^{i}M_{12} \dot{q}_{i(2)}\right)^2 \nonumber \\
		& + \,^{i}T_{22}(M_i^{-1})_{21} u_{i(1)}+\,^{i}T_{22} (M_i^{-1})_{22} u_{i(2)}. \label{equbb-01}
		\end{align}
		One can check that one solution to \eqref{equ11} is 
		\begin{equation}\label{equ12}
		^{i}T_{22} =  \sqrt{\frac{|M_i(q_i)|}{^{i}M_{11}}}.
		\end{equation}
	     So, the globally nonsingular transformation matrix $T_i(q_i)$ is obtained as follows    \begin{equation}\label{equaa-100}
		T_i(q_i)= 
		\left[
		\begin{array}{cc}
		^{i}M_{11} & ^{i} M_{12} \\
		0  &   \sqrt{\frac{|M_i(q_i)|}{^{i}M_{11}}}
		\end{array} \right].
		\end{equation}
		Consequently, the coordinate transformation \eqref{equ02} results in the partially linearized system with the state $[q_i^T, z_i^T]^T$, output $y_i$ and input $u_i=\tau_i -G_i(q_i)$
		\begin{equation}\label{equ14}
		\left\{ 
		\begin{aligned}
		&\dot{q}_i=A_i(q_i)z_i  \\
		&\dot{z}_i= f_i(q_i,z_i)+D_i(q_i) u_i  \\
		& y_i= q_i 
		\end{aligned}\right.,
		\end{equation}
		where
		$$
		A_i(q_i)=\left[ \begin{array}{cc}
		\frac{1}{\,^{i}M_{11}} & -\frac{\,^{i}M_{12}}{ \sqrt{\,^{i}M_{11} |M_i(q_i)|}} \\
		0& \sqrt{ \frac{|M_i(q_i)|}{\,^{i}M_{11}} }
		\end{array} \right],
		$$
		$$D_i(q_i)=\left[\begin{array}{cc}
		1&0 \\
		-\frac{\,^{i}M_{12}}{ \sqrt{\,^{i}M_{11} |M_i(q_i)|}} & \sqrt{ \frac{\,^{i}M_{11}}{|M_i(q_i)|} }
		\end{array} \right],
		$$
		and 
		$$
		f_i(q_i,z_i)=\left[0, -\frac{O_{i(2)}\sin(q_{i(2)})}{ \sqrt{\,^{i}M_{11} |M_i(q_i)|} \,^{i}M_{11} } z_{i(1)}^2\right]^T. 
		$$
		It can be seen that the quadratic cross terms of the unmeasurable velocities have been removed from the system dynamics \eqref{equ14}. Moreover, the matrices $A_i$ and $D_i$  are both independent of the velocity states and  bounded. Both of the above properties will facilitate the design of globally stable observers and controllers.

		\begin{remark}\label{remark-01-D}
			For future reference, denote 
			\begin{equation}\label{equ2a-1}
			\delta_{i}(q_{i(2)}) \triangleq -\frac{O_{i(2)}\sin(q_{i(2)})}{ \sqrt{\,^{i}M_{11} |M_i(q_i)|} \,^{i}M_{11} }. 
			\end{equation}
			It follows from the positive definiteness of the inertia matrix $M_i(q_i)$ that $^{i}M_{11} >0 $ and $\inf_{t} \, ^{i}M_{11}(q_i(t)) =O_{i(1)}- 2O_{i(2)}$. Since $M_i(q_i)$ is bounded, we have 
			\begin{equation}\label{equ-002} 
			\sup_{t} \delta_i[{q_{i(2)}(t)}] =\frac{O_{i(2)}}{\sqrt{k_m}} (O_{i(1)}-2O_{i(2)})^{2/3}\triangleq \bar \delta_i >0.
			\end{equation}
		\end{remark}
		
		\begin{remark}\label{remark01}
			The simplification of Euler-Lagrange systems was previously studied in \cite{MaAmVi07}\cite{Ma10}, where the conditions for the existence of the transformation matrix $T_i(q_i)$ were presented based on the equation $\dot T_i(q_i)=T_i(q_i)M_i(q_i)^{-1}C_i(q_i,\dot{q}_i)$. However, for a class of Euler-Lagrange systems, such as the robot manipulators we discussed here and unicycle-type mobile robots \cite{DoJiPa04}, such a nonsingular matrix solution $T_i(q_i)$ does not exist. So, in this paper, a wide class of transformation matrices  is derived from the relaxed  equation, i.e., $\dot{T_i}(q_i) \dot q_i=T_i(q_i)M_i(q_i)^{-1}C_i(q_i,\dot{q}_i) \dot{q}_i$ resulted from \eqref{equ03}.
		\end{remark}

		\begin{remark}\label{remark_model}
		  It can be seen that the computation of the nonsingular coordinate transformation matrix \eqref{equaa-100} relies on the exact knowledge of the inertia parameters. When the parameter uncertainties are taken into account, the construction of robust adaptive controllers needs to be considered based on the parameter linearizability property of Euler-Lagrange systems.
		\end{remark}

		\subsection{Problem Formulation}
		Consider a group of $n$ followers modeled by  \eqref{equ01}, and the leader labeled by $0$ with the same dynamics as the followers. Hence, by employing \eqref{equ02}, the leader's dynamics can also be transformed to \eqref{equ14} with the states $(q_0,z_0)$. The distributed global output-feedback tracking problem is to design local control protocols $u_i$ using only output information for all the followers, such that all the followers' states synchronize to the leader's state globally, i.e., $\lim_{t \to \infty} q_i(t)-q_0(t) =\mathbf 0, i=1,\cdots,n.$

		\section{Output-feedback tracking control} \label{sec-03}
		
		The purpose of this section is to present an observer-based control law to solve the distributed output-feedback tracking problem. Toward this end, we first design the observers to estimate the unmeasurable  velocities. 
		
		\subsection{Observer Design}\label{sec3-a}
		%
		
		Note that system \eqref{equ14} can be rewritten into the following two sub-systems  
		\begin{eqnarray}
		&&\left\{ 
		\begin{aligned}
		&\dot{q}_{i(1)} =\, ^{i}A_{11} z_{i(1)}+\, ^{i}A_{12}z_{i(2)} \\
		& \dot{z}_{i(1)} = u_{i(1)}
		\end{aligned}\right.  \label{equ33-1} \\
		&&\left\{ 
		\begin{aligned}
		& \dot{q}_{i(2)} = \, ^{i}A_{22} z_{i(2)} \\
		& \dot{z}_{i(2)}= \delta_i z_{i(1)}^2 + ^{i}D_{21} u_{i(1)} +^{i}D_{22} u_{i(2)}
		\end{aligned}\right.  \label{equ33-2}
		\end{eqnarray}
		It can be observed from \eqref{equ33-1} that the dynamics of $q_{i(1)}$ incorporate the state $z_{i(2)}$ of the second sub-system. Similarly, the dynamics of $z_{i(2)}$ also depend on the state $z_{i(1)}$ of the first sub-system in \eqref{equ33-2}. This implies that when we design the velocity observers for both of the sub-systems, the convergence analysis of the observation errors for each sub-system is still related to each other, which makes it challenging to design globally stable observers. To handle this problem, motivated by \cite{Ma10}, we aim at fully decoupling the sub-systems by constructing the new sates $x_i=[x_{i(1)},x_{i(2)}]^T$ as follows:
		\begin{align}
		&x_{i(1)}=\int_0^{q_{i(1)}} {^{i}M_{11}}(q_{i(2)}) e^{-(q_{i(1)}-s)} ds+\int_0^{q_{i(1)}} {^{i}M_{12}}(s) ds \nonumber \\
		&x_{i(2)}=\int_0^{q_{i(2)}} {^{i}M_{12}}(s)ds  \label{equ3-14}
		\end{align}
		
		Combining \eqref{equ02}, \eqref{equaa-100},  \eqref{equ33-1} and \eqref{equ33-2} and taking derivative of \eqref{equ3-14}, the dynamics of $(x,z)$ are given by 
		\begin{equation}\label{ob-02}
		\left\{
		\begin{aligned}
		&\dot{x}_{i(1)}=z_{i(1)} \\
		&\dot{z}_{i(1)}=u_{i(1)} \\
		&\dot{x}_{i(2)}=z_{i(2)} \\
		&\dot{z}_{i(2)}=\delta_i(q_{i(2)})z_{i(1)}^2+{^{i}D_{21}} u_{i(1)}+{^{i}D_{22}}u_{i(2)}
		\end{aligned} \right.,
		\end{equation}
		in which the dynamics of the first sub-system $(x_{i(1)},z_{i(1)})$ are independent of the second one $(x_{i(2)}, z_{i(2)})$. Consequently, it is relatively straightforward to 
		design the observers for the two sub-systems in \eqref{ob-02}. 
		For the first sub-system, the observer is designed as 
		\begin{equation}\label{ob-03}
		\left\{ 
		\begin{aligned}
		\dot{\hat{x}}_{i(1)}=-k_{o,1}(\hat{x}_{i(1)}-x_{i(1)}) +\hat{z}_{i(1)} \\
		\dot{\hat{z}}_{i(1)}=-k_{o,2} (\hat{x}_{i(1)}-x_{i(1)})+u_{i(1)}
		\end{aligned}
		\right.,
		\end{equation}
		where $\hat{x}_i$ and $\hat{z}_i$ are the observations of $x_i$ and $z_i$, respectively. Here, $k_{o,1}$ and $ k_{o,2} $ are positive observer gains. Correspondingly, the observation errors are defined as $\tilde{x}_i=\hat{x}_i-x_i$ and $\tilde{z}_i=\hat{z}_i-z_i$, whose dynamics are of the form
		%
		{\setlength \arraycolsep{4pt}
			\begin{equation} \label{ob-05}
			\left [ \begin{array}{c}
			\dot{\tilde{x}}_{i(1)} \\
			\dot{\tilde{z}}_{i(1)}
			\end{array}\right]=
			\left [ \begin{array}{cc}
			-k_{o,1} & 1 \\
			-k_{o,2} & 0
			\end{array}\right ] \left [ \begin{array}{c}
			\tilde{x}_{i(1)} \\
			\tilde{z}_{i(1)}
			\end{array}\right] \dfb \tilde{A}  \left [ \begin{array}{c}
			\tilde{x}_{i(1)} \\
			\tilde{z}_{i(1)}
			\end{array}\right].
			\end{equation}
			It can be easily checked that matrix $\tilde{A}$ is Hurwitz, and therefore system \eqref{ob-05} is exponentially stable at the origin. So  
			\begin{equation}\label{ob-06}
			\lim_{t\to \infty} \left [ \begin{array}{c}
			\tilde{x}_{i(1)}(t) \\
			\tilde{z}_{i(1)}(t)
			\end{array}\right] = 0.
			\end{equation}}
		
		For the second sub-system $(x_{i(2)}, z_{i(2)})$, the  observer is constructed as
		\begin{equation}\label{ob-08}  
		\left\{ 
		\begin{aligned}
		&\dot{\hat{x}}_{i(2)} =-k_{o,1} (\hat{x}_{i(2)}-x_{i(2)}) +\hat{z}_{i(2)} \\
		&\dot{\hat{z}}_{i(2)}=-k_{o,2}(\hat{x}_{i(2)}-x_{i(2)})+ \delta_i\hat{z}_{i(1)}^2+{^{i}D_{2j}}u_{i(j)}
		\end{aligned}
		\right.,
		\end{equation}
		where $j=1,2$. 
		In view of  \eqref{ob-02} and \eqref{ob-08}, we have 
			\begin{equation} \label{ob-12}
			\left [ \begin{array}{c}
			\dot{\tilde{x}}_{i(2)} \\
			\dot{\tilde{z}}_{i(2)}
			\end{array}\right]=
			\tilde{A} \left [ \begin{array}{c}
			\tilde{x}_{i(2)} \\
			\tilde{z}_{i(2)}
			\end{array}\right]+  \left [ \begin{array}{c}
			0 \\
			h_i(t)\tilde{z}_{i(1)}
			\end{array}\right] ,
			\end{equation}
		where $h_i(t)=\delta_i(q_{i(2)}(t)) (\tilde{z}_{i(1)}+2{z}_{i(1)}(t))$ 
		 is continuous in $t$ and $\tilde z_{i(1)}$, and locally Lipschitz in $\tilde z_{i(1)}$. Note that both \eqref{ob-05} and the nominal part of \eqref{ob-12} are uniformly globally exponentially stable (UGES). Then, the origin of the cascaded system \eqref{ob-05} and \eqref{ob-12} is UGES \cite{PaLeLo98}, namely, $\tilde{x}$ and $\tilde{z}$ uniformly globally exponentially converge to zero.

		\subsection{Observer-Based Control Law Design}\label{sec-04}
		The following assumptions are made throughout this paper. 
		\begin{assumption}\label{assum-01}
			The leader's state information $(x_0(t), z_0(t))$ satisfies $\sup_t \| \dot{z}_0(t)\| \leq \bar z_0$. 
		\end{assumption}
		
	          \begin{assumption}\label{assum-02}
				The communication relationships among the $n+1$ robots form a directed graph $\bb G$ that contains a spanning tree rooted at the leader.  
			\end{assumption}
		
		In order to keep this paper self-contained, two lemmas are presented. 
					
				\begin{lemma}\cite{La05} \label{rd2_lem01}
					Let $A\in \R^{m\times n}, B\in \R^{r\times s}, C\in \R^{n\times p}, D\in \R^{s\times t} $. Then
					\begin{equation}\label{eq_r2_la}
						(A\otimes B)(C \otimes D)=AC \otimes BD \quad (\in \R^{mr \times pt}).
					\end{equation}
					And, for all $A$ and $B$, 
					\begin{equation}\label{eq_r2_lb}
						(A\otimes B)^T =A^T \otimes B^T.
					\end{equation}	
				\end{lemma}
				
				\begin{lemma}\cite{Qu09,ZhLe12}\label{rd2_lem02}
					Under Assumption \ref{assum-02}, $( L+  B)$ is a nonsingular M-matrix. Define 
					\begin{equation} \label{eq_lem2}
						\begin{aligned}
							& H=[h_1,\cdots,h_n]^T=( L+ B)^{-1} \mathbf 1_n  \\
							& P=\diag \{p_i\}=\diag \{1/h_i\}
						\end{aligned}
					\end{equation}					
					Then $P$ is positive definite and the matrix $Q$ defined as 
					\begin{equation}\label{eq_lem2_1}
						Q=P( L+ B)+( L+ B)^TP
					\end{equation}  
					is also positive definite.
				\end{lemma}

		To come up with the observer-based distributed control laws, an auxiliary variable is introduced as follows:
		\begin{equation}\label{cd-02}
		\xi_i=\hat{z}_i-z_0+ \kappa (x_i-x_0),
		\end{equation}
		where $\kappa>1 $ is a constant. The local differences are defined as
		\begin{equation}\label{cd-03}
		\left\{
		\begin{aligned}
		&x_{ir}=\sum_{j \in \scr N_i} a_{ij} (x_i-x_j)+b_i(x_i-x_0) \\
		&z_{ir}=\sum_{j \in \scr N_i} a_{ij}(\hat{z}_i-\hat{z}_j)+b_i(\hat{z}_i-z_0) 
		\end{aligned}\right.,
		\end{equation}
		and
		\begin{equation}\label{cd-04}
		s_i=z_{ir}+\kappa x_{ir}.
		\end{equation}
		 The auxiliary variable $s_i$ can be written into a compact form
		\begin{equation}\label{cd-05}
		s=\left(( L+ B)\otimes I_2\right) \xi.
		\end{equation}
		The distributed control law for robot $i$ is proposed as follows
				\begin{equation} \label{eq_r2_in}
				u_i = D_i^{-1}(q_i)\left(k_{c,1}\tilde{x}_i- k_{c,2} s_i -k_{c,3} {\rm sign}(s_i)-f_i(q_i,\hat{z}_i)\right),
				\end{equation}
				where ${\rm sign} (s_i)=\left[ {\rm sign}(s_{i(1)}),  {\rm sign}(s_{i(2)})\right]^T \in \R^{2}$. Here, $k_{c,1}$ can be any positive  number, and $k_{c,2}$ and $k_{c,3}$ are positive numbers satisfying	
				{\small			
					\begin{equation} \label{eq_r2_them}
					\left\{ 
					\begin{aligned}
					&   k_{c,2}> \frac{1}{\underline{\lambda}_Q} \left(3\kappa \overline{\lambda}_P+ \kappa^2 \overline{\lambda}_P+ \frac{\kappa^2\overline{\lambda}_P}{2(\kappa -1)}+|k_{c,1}-k_{o,2}|\overline{\lambda}_{P} \overline{\sigma}_{( L + B)}\right)\\
					& k_{c,3}> \bar z_0
					\end{aligned} \right. 
					\end{equation}}
     where the real symmetric matrices $P$ and $Q$ are defined  in \eqref{eq_lem2} and \eqref{eq_lem2_1}, respectively.

		
		
		\section{Main results} \label{sec-04}
			The main result of this paper is given below. 
			\begin{theorem} \label{theom01}
				Under Assumptions \ref{assum-01} and \ref{assum-02}, consider the system \eqref{ob-02} transformed from the mechanical system \eqref{equ01} in closed loop with the observer-based controllers given by \eqref{eq_r2_in}. Then the origin of the closed-loop system is UGES for the control gains satisfying \eqref{eq_r2_them} and any positive observer gains $k_{o,1}$ and $k_{o,2}$.
			\end{theorem}

				\begin{proof}[Proof of  Theorem \ref{theom01}]
						
			The Lyapunov function candidate is chosen as 
		\begin{align} \label{eq_r2_001}
			V_c&=\frac{1}{2}s^T(P\otimes I_2) s+ \frac{\varpi}{2}x_r^T x_r \nonumber \\
			&= \frac{1}{2}
			 \left[
		   \begin{array}{c}
		    s \\
		    x_r
		   \end{array}	
			\right]^T  \left[ 
			\begin{array}{cc}
			P\otimes I_2 & 0\\
			0 & \varpi I_{2n}
			\end{array}
			 \right] \left[
			 \begin{array}{c}
			 s \\
			 x_r
			 \end{array}	
			 \right] \nonumber \\
			 & \dfb \frac{1}{2} y^T \bar P y
			\end{align}
			where $y \dfb [s^T, x_r^T]^T$, and $\varpi$ is a positive scalar satisfying $\varpi > \kappa^2 \overline{\lambda}_P/2(\kappa-1)$. It is straightforward to check that matrix $\bar P$ is positive definite and 
   \begin{equation}\label{eq_r2_15}
   \frac{\underline \lambda_{\bar P}}{2} \|y\|^2 \leq V_c(y) \leq  \frac{\overline \lambda_{\bar P}}{2} \|y\|^2.
   \end{equation}			
		The generalized derivative of $V_c$ (see \citep[Remark 3.7]{GhMeRe16}) is given by 
	  	\begin{align}\label{eq_r2_01}
				\dot V_c &= s^T (P\otimes I_2) (( L+ B)\otimes I_2) \big[ \dot{\hat z}- (\mathbf 1_n \otimes \dot z_0) \nonumber \\
				&+\kappa z-\kappa (\mathbf 1_n \otimes z_0) \big]+\varpi x_r^T (s-\kappa x_r-\tilde z_r), 
				\end{align}
			where  $\tilde{z}_{r}=(( L+ B)\otimes I_2) \tilde{z}$. Note that from \eqref{ob-03} and \eqref{ob-08}, we know
			\begin{equation}\label{eq_r2_02}
			\dot{\hat z}=-k_{o,2} \tilde x+f(q,\hat z)+D(q) u.
			\end{equation}
		Also, the control input \eqref{eq_r2_in} can be written in its stacked form as
			\begin{equation}\label{eq_r2_10}
			u=D(q)^{-1} \left(k_{c,1} \tilde x - k_{c,2} s-k_{c,3} {\rm sign} (s)-f(q,\hat z)\right),
			\end{equation}
			where $D(q)^{-1}= {\rm blockdiag} \{D_1^{-1},\cdots, D_n^{-1} \}\in \R^{2n\times 2n}$ and ${\rm sign} (s)$$=\left[ {\rm sign}(s_1)^T , \cdots, {\rm sign}(s_n)^T\right]^T \in \R^{2n}$.		
			Substituting \eqref{eq_r2_02}  and \eqref{eq_r2_10} into \eqref{eq_r2_01}, we have  
				\begin{align}\label{eq_r2_03}
				\dot V_c=& -k_{c,2}s^T(P( L+ B)\otimes I_2)s-\varpi \kappa x_r^T x_r\nonumber \\
			&	+(k_{c,1}-k_{o,2})s^T(P( L+ B)\otimes I_2)\tilde{x} \nonumber \\				
				& + \kappa s^T (P( L+ B)\otimes I_2) (z-\hat z +\hat z- \mathbf 1_n \otimes z_0)  \nonumber \\
			&	+\varpi x_r^Ts-\varpi x_r^T \tilde z_r -k_{c,3}  s^T \left(P  L\otimes I_2 \right) {\rm sign(s)}\nonumber \\
				& - k_{c,3} s^T \left(P B \otimes I_2 \right) {\rm sign(s)} + s^T (P B\otimes I_2)(\mathbf 1_n \otimes \dot z_0),    
				\end{align}
			where  Lemma \ref{rd2_lem01} and the equality that $( L \otimes I_2) (\mathbf 1_n \otimes \dot z_0)=\mathbf 0$ have been used. 
			 
             Note that 
			\begin{equation}\label{eq_r2_04}
			z-\hat z=-\tilde{z}=- (( L+ B)\otimes I_2)^{-1} \tilde z_r,
			\end{equation}
			and 
			\begin{equation}\label{eq_r2_05}
			\hat z-\mathbf 1_n \otimes z_0= (( L+ B)\otimes I_2)^{-1} z_r=(( L+ B)\otimes I_2)^{-1} (s-\kappa x_r).
			\end{equation}
			Then, substituting \eqref{eq_r2_04} and \eqref{eq_r2_05} into \eqref{eq_r2_03} yields
				\begin{align}\label{eq_r2_06}
				\dot V_c = &-k_{c,2}s^T(P( L+ B)\otimes I_2)s -\varpi \kappa x_r^T x_r\nonumber \\
				& - \kappa s^T (P( L+ B)\otimes I_2)(( L+ B)\otimes I_2)^{-1} \tilde z_r  \nonumber \\
				&  + \kappa s^T (P( L+ B)\otimes I_2) (( L+ B)\otimes I_2)^{-1} (s-\kappa x_r)\nonumber  \\
				&  +(k_{c,1}-k_{o,2})s^T(P( L+ B)\otimes I_2)\tilde{x}+\varpi x_r^Ts-\varpi x_r^T \tilde z_r \nonumber \\
				& -k_{c,3}  s^T \left(P  L\otimes I_2 \right) {\rm sign(s)}- k_{c,3} s^T \left(P B \otimes I_2 \right) {\rm sign(s)} \nonumber \\
				& + s^T (P B\otimes I_2)(\mathbf 1_n \otimes \dot z_0).
				\end{align}
			Therefore 
				\begin{align}\label{eq_r2_07}
				\dot V_c =& -k_{c,2}s^T(P( L+ B)\otimes I_2)s -\varpi \kappa x_r^T x_r - \kappa s^T (P\otimes I_2) \tilde z_r \nonumber  \\
				& + \kappa s^T (P\otimes I_2) (s-\kappa x_r) +\varpi x_r^Ts-\varpi x_r^T \tilde z_r \nonumber \\
				&+(k_{c,1}-k_{o,2})s^T(P( L+ B)\otimes I_2)\tilde{x} 	 \nonumber \\
					& -k_{c,3}  s^T \left(P  L\otimes I_2 \right) {\rm sign(s)}- k_{c,3} s^T \left(P B \otimes I_2 \right) {\rm sign(s)} \nonumber \\
				& + s^T (PB\otimes I_2)(\mathbf 1_n \otimes \dot z_0)\nonumber \\
			  =& -\frac{k_{c,2}}{2} s^T(Q\otimes I_2) s-\varpi\kappa x_r^T x_r+\kappa s^T (P\otimes I_2)s \nonumber \\
				& -\kappa s^T (P\otimes I_2) \tilde z_r -\kappa^2 s^T (P \otimes I_2) x_r +\varpi x_r^Ts-\varpi x_r^T \tilde z_r\nonumber \\
	  &	+(k_{c,1}-k_{o,2})s^T(P( L+ B)\otimes I_2)\tilde{x} 	 \nonumber \\
		& -k_{c,3}  s^T \left(P  L\otimes I_2 \right) {\rm sign(s)}- k_{c,3} s^T \left(P B \otimes I_2 \right) {\rm sign(s)} \nonumber \\
			              & + s^T (P B\otimes I_2)(\mathbf 1_n \otimes \dot z_0).				
				\end{align}
			
			From Lemma \ref{rd2_lem02}, we know matrix $Q$ and $P$ are positive definite. Then it follows 
				\begin{align}\label{eq_r2_08}
				\dot V_c \leq & -\frac{k_{c,2}}{2} \underline{\lambda}_Q \|s\|^2-\varpi \kappa \|x_r\|^2 + \kappa \overline{\lambda}_P \|s\|^2 \nonumber \\ 
				&+\frac{\kappa}{2} \overline{\lambda}_P(\|s\|^2+\|\tilde{z}_r\|^2)+\frac{\kappa^2}{2}\overline{\lambda}_P(\|s\|^2+\|x_r\|^2)\nonumber \\
				 &+\frac{\varpi}{2}(\|s\|^2+\|x_r\|^2)+\frac{\varpi}{2}(\|\tilde{z}_r\|^2+\|x_r\|^2)\nonumber \\	
				&+\frac{|k_{c,1}-k_{o,2}|}{2} \overline{\lambda}_{P} \overline{\sigma}_{( L + B)}(\|s\|^2+\|\tilde{x}\|^2)
				 \nonumber \\
				& -k_{c,3} \sum_{i=1}^n p_i b_i \|s_i \|_1+ \sum_{i=1}^{n} p_i b_i \bar z_0 \|s_i\|  \nonumber \\
				\leq & -\alpha_1 \|s\|^2 -\alpha_2 \|x_r\|^2+\alpha_3 \|\tilde{z}_r\|^2 +\alpha_4 \|\tilde x\|^2, 
			\end{align}
			where $\alpha_i, i=1,\cdots,4$, are given by 
			{\small
			\begin{align}
			& \alpha_1 = \frac{1}{2} \left(k_{c,2}\underline{\lambda}_Q -3\kappa \overline{\lambda}_P- \kappa^2\overline{\lambda}_P-\varpi-|k_{c,1}-k_{o,2}|\overline{\lambda}_{P} \overline{\sigma}_{( L + B)} \right),\nonumber \\
			& \alpha_2= \varpi\kappa-\varpi-\frac{\kappa^2}{2} \overline{\lambda}_P, \nonumber \\
			& \alpha_3 = \frac{\kappa}{2} \overline{\lambda}_P +\frac{\varpi}{2}, \nonumber \\
			& \alpha_4 =\frac{|k_{c,1}-k_{o,2}|}{2} \overline{\lambda}_{P} \overline{\sigma}_{( L + B)}.
			\end{align}	}
 $\overline{\sigma}_{(X)}$ represents the largest singular value of matrix $X$, and here we have used the facts that $s^T\left(P  L \otimes I_2 \right) {\rm sign(s)} \geq 0 $ and  $\|s_i\|_1 \geq \|s_i \|$. Under the condition \eqref{eq_r2_them} 
 and the constraint for $\varpi$ in \eqref{eq_r2_001},
 the parameters $\alpha_i>0, i=1,\cdots,4$. Hence, $\dot V_c$  satisfies
		\begin{align}\label{eq_r2_14}
		\dot V_c \leq - \min\{\alpha_1, \alpha_2 \} \|y\|^2+\alpha_3 \|\tilde{z}_r \|^2 +\alpha_4 \|\tilde x\|^2. 
		\end{align}
		Combining \eqref{eq_r2_15} and \eqref{eq_r2_14}, we get
		\begin{align}\label{eq_r2_16}
		\dot V_c(y) &\leq - \frac{2 \min\{\alpha_1, \alpha_2 \}}{\overline \lambda_{\bar P}} V_c(y) +\alpha_3 \|\tilde{z}_r \|^2 +\alpha_4 \|\tilde x\|^2.  
		\end{align}
		
	 Recall that $\|\tilde{z}\|$ and $\|\tilde{x}\|$ converge to zero exponentially from Section \ref{sec3-a}, so does $\|\tilde z_r\|$ due to the fact that  $\tilde{z}_{r}=(( L+ B)\otimes I_2) \tilde{z}$. Then from the Converse Theorem \cite{Khalil02}, there exist a function $V_o(t,\tilde z_r,\tilde x)$ that satisfies the inequalities
			\begin{equation}\label{eq-f-1}
				\dot{V}_o \leq -\beta( \|\tilde z_r\|^2 + \|\tilde x\|^2 ),	
			\end{equation}
		where $\beta$ is a positive constant. Then, we choose the overall Lyapunov function candidate as 
		\begin{equation}\label{eq-f-2}
		    V=V_c+\varrho V_o,
		\end{equation}
		where $\varrho$ is a positive constant satisfying $\varrho > \max\{\frac{\alpha_3}{\beta}, \frac{\alpha_4}{\beta}\}$. For simplicity, we use $\bar \alpha \dfb 2 \min\{\alpha_1, \alpha_2 \}/\overline \lambda_{\bar P}$. Taking the time derivative of both sides of \eqref{eq-f-2}, and combining \eqref{eq_r2_16} and \eqref{eq-f-1}, we have
		\begin{equation}
		    \dot V \leq  -\bar \alpha V_c - \max\{(\beta \varrho - \alpha_3), (\beta \varrho - \alpha_4) \} (\|\tilde z_r \|^2 + \|\tilde x \|^2).
		\end{equation}
	Denoting $\bar \beta \dfb \max\{(\beta \varrho - \alpha_3), (\beta \varrho - \alpha_4) \}$, one has 
	 \begin{equation} 
	   \dot V \leq \left\{ \begin{aligned}
	      & - \bar \alpha V - (\bar \beta - \bar  \alpha \varrho) V_o,  \;\; \text{if} \;\; \bar \beta > \bar \alpha \varrho, \\
	      & -\frac{\bar \beta}{\varrho} V -(\bar \alpha - \frac{\bar \beta}{\varrho} ) V_c,  \;\; \text{if} \;\; \bar \beta \leq \bar \alpha \varrho.
	   \end{aligned}
	   \right.
	 \end{equation} 		
		Consequently, we get 
		\begin{equation}
		   \dot V \leq - \min\{ \bar \alpha, \bar \beta /\varrho \} V.
		\end{equation}
    It then follows from \citep[Theorem 4.10]{Khalil02} that $[s^T, x_r^T]^T =\mathbf 0$ is UGES,
	 which implies $ x_{ir}$ and $z_{ir}$ converge to $\mathbf{0}$ according to the definitions \eqref{cd-03} and \eqref{cd-04}. By invoking \citep[Theorem 4]{KhXiMa09} and  $\lim_{t\to \infty} \hat{z}=z$, we know $z_i$ and $x_i$, respectively, converge to $z_0$ and $x_0$ in the sense of uniform global exponential stability. This implies $q_i(t)$ exponentially globally converges to $q_0(t)$ due to the same integrand in definition \eqref{equ3-14}.
	
        Till now, it has been proved that the observation errors uniformly globally exponentially converge to zero for any positive observer gains $k_{o,1}$ and $k_{o,2}$, and the origin of the closed-loop system is UGES for the control gains satisfying \eqref{eq_r2_them}. In addition, the global convergence of the tracking errors can still be guaranteed if the observer states $(\hat{x}, \hat{z})$ are replaced by the real states $(x, z)$ in controller \eqref{eq_r2_in}. These imply that the observer and the controller can be designed separately, namely the separation principle holds. 
		 \end{proof}

			\begin{remark}
			 In comparison with the $I \& I$ technique reported in \cite{AsOrVe10,StAaKa11,SaNuKi12}, our proposed approach can reduce the complexity of the observer and the distributed controller design. Moreover, the theoretical analysis can also be more easily carried out based on the resulted partially linearized cascaded system. 
				The method to partially linearize the class of Euler-Lagrange system can be applied to deal with the systems of $n$ DOFs,  $n>2$, by setting the transformation matrix in a particular upper triangular form. However, to solve a set of PDEs with high order is a challenging, yet interesting issue for future research. 
		\end{remark}

\begin{remark}
	To eliminate the chattering behavior caused by the  signum function, in practice, various continuous functions such as the hyperbolic tangent function and the saturation function \cite{SoLi91} have been employed to approximate the discontinuous signum function.  
	Although the damage exerted on the actuator caused by  discontinuity could be avoided, the states of the closed-loop system might only be stabilized within a bounded neighborhood of the equilibrium, instead of converging to the equilibrium. 	
\end{remark}

		\section{simulations} \label{sec-05}
		To validate the theoretical results derived in the preceding sections, we shall consider four robot manipulators modeled by \eqref{equ01}, with the physical parameters taken from \cite{Re09} as $[m_1, m_2, l_1, l_2,l_{c1},$ $ l_{c2},J_1,  J_2,g]=[0.5, 0.4, 0.4, 0.3, 0.2, 0.15, 0.0067, 0.003, 9.8]$. The communication topology among the four followers and the leader is represented by Fig. \ref{figure01}.
		
		The initial values of $x_i(t)$ and $\hat z_i(t),  i=1,\cdots,4,$ are set as $3*\text{rands}(2,1)$ and $[0,0]^T$, respectively,  and the leader's trajectory is $[2*t,\sin(t)]^T$, satisfying Assumption \ref{assum-01}. The following parameters are used in the simulation. Observer parameters: $k_{o,1}=3, k_{o,2}=5$, control gains: $k_{c,1}=5, k_{c,2}=6, k_{c,3}=3$.  By employing the distributed control laws \eqref{eq_r2_in}, the numerical simulation  results are shown in Fig. \ref{figure02}-Fig. \ref{figure04}.  It can be seen from Fig. \ref{figure02} and Fig. \ref{figure03} that the tracking errors for each event converge to zero. Fig. \ref{figure04} indicates that each agent can precisely observe the unmeasurable velocities using the proposed observer \eqref{ob-03} and \eqref{ob-08}.
		
		\usepgflibrary{arrows}
		
		\tikzset{
			agent/.style = {circle, draw=black, fill=lightgray},
			cable/.style={ thick, densely dashed},
			strut/.style={line width=1.5pt},
			bar/.style={ line width=0.5pt}
		}
		
		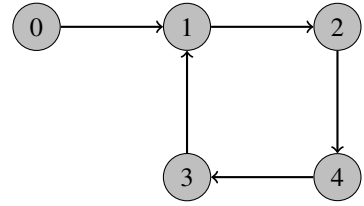
\begin{figure}
			\centering 
			\setlength{\abovecaptionskip}{-0pt}
			\begin{tikzpicture}
			\node[agent] (0) at(-2,2) {0};
			\node[agent] (1) at(0,2) {1};
			\node[agent] (2) at(2,2) {2};
			\node[agent] (3) at(0,0) {3};
			\node[agent] (4) at(2,0) {4};
			\draw[->,thick] (0) -- node[]{} (1);
			\draw[->,thick] (1) -- node[]{} (2);
			\draw[<-,thick] (1) -- node[]{} (3);
			\draw[->,thick] (2) -- node[]{} (4);
			\draw[<-,thick] (3) -- node[]{} (4);
			\end{tikzpicture}
			\caption{Interaction topology, where agent $0$ is the leader.}
			\label{figure01}
		\end{figure}
		\vspace{-10pt}
		
%
%

		\begin{figure}[!htb]
	\centering
	\setlength{\abovecaptionskip}{0pt}
	\includegraphics[width=\hsize]{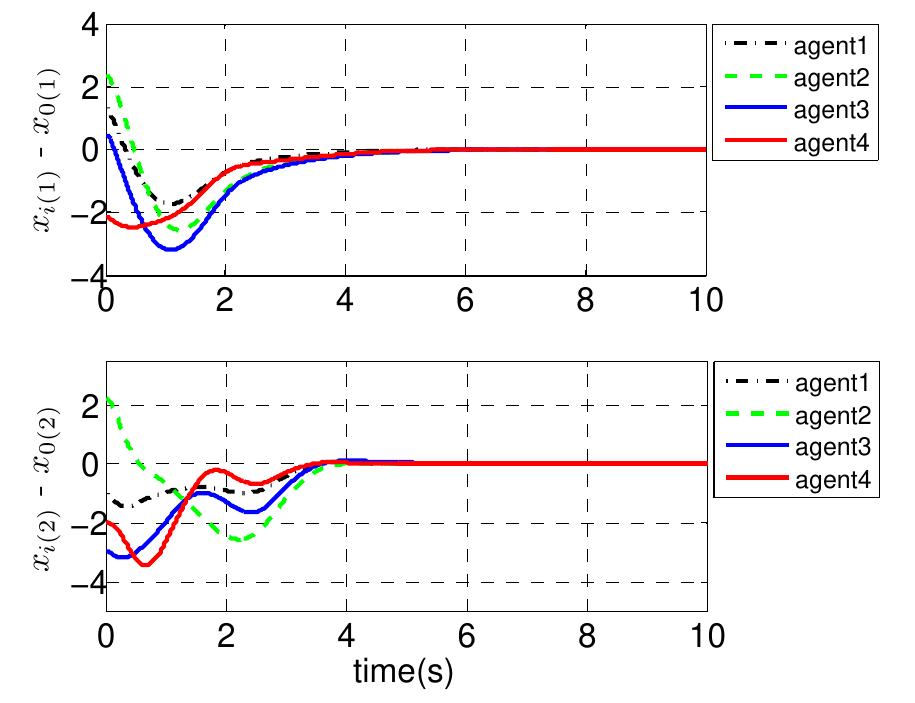}
	\caption{Position tracking errors for each agent}  
	\label{figure02}
\end{figure}

\begin{figure}[!htb]
	\centering
	\setlength{\abovecaptionskip}{0pt}
	\includegraphics[width=\hsize]{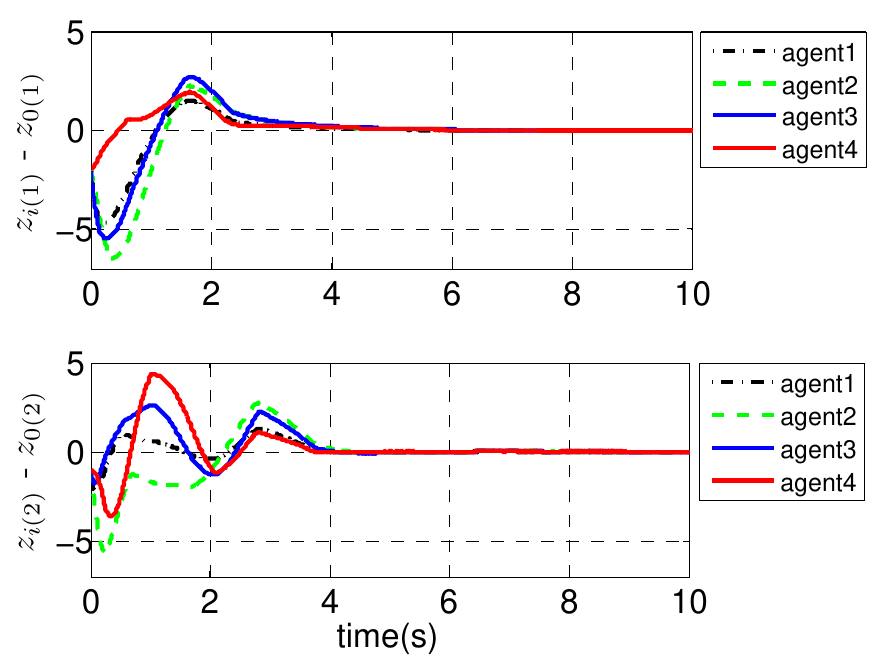}
	\caption{Velocity tracking errors for each agent }  
	\label{figure03}
\end{figure}

\begin{figure}[!htb]
	\centering
	\setlength{\abovecaptionskip}{0pt}
	\includegraphics[width=\hsize]{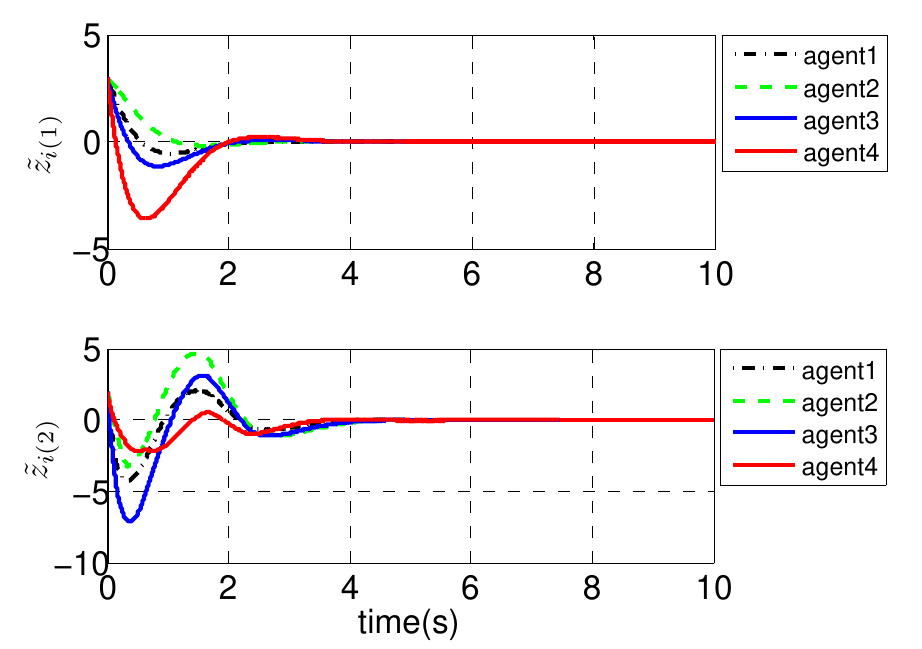}
	\caption{Velocity observation errors for each agent }  
	\label{figure04}
\end{figure}
		
		\section{conclusion} \label{sec-06}
		In this paper, for a class of Euler-Lagrange systems that cannot be fully linearized by output-feedback, we have constructed a nonsingular coordinate transformation matrix to partially linearize the nonlinear systems. Then, observers have been designed to overcome the unavailability of the velocity measurements. We have also proposed observer-based control laws by output-feedback  such that the followers  uniformly globally exponentially track the leader.  It should be noted that the system discussed here is fully actuated. Future research directions may include distributed global output-feedback control for a class of underactuated Euler-Lagrange systems, such as nonholonomic wheeled mobile robots studied in \cite{DoJiPa04}. We are also interested in distributed observerless global output-feedback control for Euler-Lagrange systems with parameter uncertainties reported recently in \cite{Lo16}.

\section*{Acknowledgment}

	The authors sincerely thank the associate editor and all the
	anonymous reviewers for their valuable comments which helped
	improve the quality of this paper. 
	

\ifCLASSOPTIONcaptionsoff
  \newpage
\fi



\bibliographystyle{IEEEtran}
\bibliography{ref_qingkai}
\title{Errata to ``Distributed Global Output-Feedback Control for a Class of Euler-Lagrange Systems''}
%
%
%


\author{Qingkai Yang}	
\date{}
%
\maketitle

In the above paper \cite{Yang2017TAC}, the expression of the matrix $A_i(q_i)$ in (10) should write as 
\setcounter{equation}{9}
\begin{equation}
A_i(q_i) = \left[ \begin{array}{cc}
\frac{1}{^i M_{11}} & - \frac{ ^i M_{12}}{\sqrt{ ^i M_{11} |M_i(q_i)|}} \vspace{5pt}  \\ 
0 & \sqrt{ \frac{ ^i M_{11}}{|M_i(q_i)|}}
\end{array}\right]
\end{equation}

The equation (15) is corrected to the following form
\setcounter{equation}{14}
\begin{equation}
\begin{aligned}
x_{i(1)} &= \int \, ^i M_{11}(q_{i(2)}) d q_{i(1)} + \int_{0}^{q_{i(2)}} \, ^i M_{12}(s)ds \\
x_{i(2)} &= \int_{0}^{q_{i(2)}} \, ^i T_{22}(s) ds
\end{aligned}
\end{equation}

Note that given a function $f$, a sufficient condition for the existence of a primitive function is that $f$ is continuous.
Therefore, the existence of $x_{i(1)}$ can be guaranteed by taking the continuity of the term $^i M_{11}(q_{i(2)}) \dot q_{i(1)}$ into account.


\bibliographystyle{IEEEtran}

%

\end{document}